\documentclass[letterpaper, 10 pt, conference]{ieeeconf} 

\IEEEoverridecommandlockouts                              
 
\usepackage[utf8]{inputenc}  
  
\usepackage{mathtools}
\mathtoolsset{showonlyrefs,showmanualtags}
\usepackage{amssymb,amsthm}
\usepackage{dsfont}
\usepackage{booktabs} 
\usepackage{color}
\usepackage{bbold} 
\usepackage[yyyymmdd,hhmmss]{datetime}
\usepackage{bm,upgreek}
\usepackage{tablefootnote}
\usepackage{nth}
\usepackage{pgfplots}
\usepackage{tikz}
\usetikzlibrary{shapes,snakes}
\usetikzlibrary{patterns}
\usepackage{multirow}
\usepackage{footnote}
\usetikzlibrary{hobby}

\usepackage{algorithm}
\usepackage{textcomp}
\usepackage{xcolor}
\usepackage{multicol}
\usepackage{algorithmic}
\def\BibTeX{{\rm B\kern-.05em{\sc i\kern-.025em b}\kern-.08em
    T\kern-.1667em\lower.7ex\hbox{E}\kern-.125emX}}

\newtheorem{assumption}{Assumption}
\newtheorem{theorem}{Theorem}
\newtheorem{definition}{Definition}
\newtheorem{remark}{Remark}
\newtheorem{lemma}{Lemma}

\newtheorem{problem}{Problem}

\DeclareMathOperator{\DIAG}{diag}
\DeclareMathOperator{\VEC}{vec}
\DeclareMathOperator{\ROW}{row}
\newcommand{\NF}[1]{\| #1\|_F}  

\DeclareMathOperator{\VAR}{var}

\newcommand{\ID}[1]{ \mathbb{1} (#1 ) }

\newcommand{\Exp}[1]{\mathbb{E}\big[ #1\big]}

\newcommand{\Compress}{\medmuskip=0mu
\thinmuskip=0mu
\thickmuskip=0mu}

\title{\LARGE \bf
Reinforcement Learning in Linear Quadratic Deep Structured Teams:  Global Convergence of Policy Gradient  Methods
}

\author{Vida Fathi, Jalal Arabneydi  and Amir G. Aghdam
\thanks{This work has been supported by the Natural Sciences and Engineering Research Council of Canada (NSERC) under Grant RGPIN-262127-17.}  
\thanks{ Vida Fathi, Jalal Arabneydi and Amir G. Aghdam are with the  Department of Electrical and Computer Engineering, 
        Concordia University, 1455 de Maisonneuve Blvd. West, Montreal, QC, Canada, Postal Code: H3G 1M8.  Email:  {\tt\small v\_fathi@encs.concordia.ca}, {\tt\small jalal.arabneydi@mail.mcgill.ca},   
        {\tt\small aghdam@ece.concordia.ca}}%
}

\begin{document}
\maketitle

\vspace*{-5cm}{\footnotesize{Proceedings of IEEE  Conference on Decision and Control, 2020.}}
\vspace*{4.2cm}

\thispagestyle{empty}
\pagestyle{empty}
\begin{abstract}
In this paper, we  study the global convergence of model-based and  model-free  policy  gradient descent and natural  policy gradient descent algorithms for linear quadratic deep structured teams. In  such systems,    agents are partitioned into a few  sub-populations wherein the agents in each sub-population are  coupled in the dynamics and cost function  through a  set  of  linear regressions of the states and actions of all agents. Every agent observes its local state and the linear regressions of states, called deep states.  For a sufficiently small  risk factor and/or sufficiently large population, we prove  that model-based policy gradient methods globally  converge to the  optimal solution.  Given an arbitrary  number of agents, we develop model-free policy gradient and natural policy gradient  algorithms  for the special case of risk-neutral cost function.  The proposed algorithms are scalable with respect to the number of agents due to the fact that the dimension of  their policy  space is independent of the number of agents in each sub-population. Simulations are  provided to verify the theoretical results. 
\end{abstract}


\section{Introduction}
In today's world, networked control systems are ubiquitous, ranging from smart grids and economics to   communication networks and epidemics. Such systems often consist of  many  decision makers (nodes)  with complex interactions.  In general,  finding an optimal (or even sub-optimal) solution  in  networked control systems  is  difficult. This difficulty exacerbates when  practical restrictions are  taken into account such as    limited  number of  computation and communication  resources  and  incomplete  knowledge of the model.

It is well known that the number of computational elements (such as memory and time) increases exponentially with the number of decision makers in stochastic dynamic control systems. In addition,  the  lack of  centralized  communication among the decision makers can  lead to different perspectives at the agent level,  where solving  a simple linear quadratic problem is challenging; see~\cite{Witsenhausen1968Counterexample} for  a counterexample in which  the resultant optimization problem is non-convex.  Furthermore,   the above  challenges are worsened when the underlying  model is not completely known. Therefore, it is of special interest in control theory  to find a class of  models in which the above challenges  can be addressed to some extent.

Motivated from  recent  developments  in artificial intelligence,  deep structured teams and games have been introduced in~\cite{Jalal2019MFT,Jalal2019Automatica,Jalal2019risk,Jalal2020Nash,Jalal2020CCTA,Masoud2020CDC}, 
 which may be viewed as the generalization of mean-field teams proposed  in~\cite{arabneydi2016new} and showcased in~\cite{Jalal2017linear,JalalCDC2015,JalalCDC2018,
 Jalal2019LCSS,JalalCCECE2018,JalalACC2018}.  In such systems, the interaction between the decision makers  is modelled by a set of   linear regressions (weighted averages) of the states and actions,  where the weights represent  the dominant features of the model. We call such models deep structured  because  the interaction between the  decision makers  is similar to  that  between the neurons of a deep feed-forward  neural network.
 
In this paper, we  study  the global convergence of the model-based and model-free  policy gradient descent and natural policy gradient descent algorithms.  Since the convexity in action space does not imply the convexity in policy space~\cite{fazel2018global}, we use the notions of gradient domination and locally Lipschitz continuity to show that   the policy gradient descent and natural policy gradient descent algorithms converge to the globally optimal solution. 
 In~\cite{Luo2019} and~\cite{lu2020decentralized}, the authors propose several policy gradient (reinforcement learning) algorithms for the special  case of homogeneous weights with the  risk-neutral cost function.  In  this paper, however, we consider a more general setup with risk-sensitive cost function and deep structured model consisting of multiple features and heterogeneous weights.

The remainder of the paper is organized as follows. The problem is formulated in Section~\ref{sec:problem} and the main  results are presented in Sections~\ref{sec:main} and~\ref{sec:main2}. Three types of implementation are discussed in Section~\ref{sec:implement}.
To verify the obtained results, some simulations are provided in Section~\ref{sec:numerical}. Finally, the paper is concluded in Section~\ref{sec:conclusion}. 
\section{Problem Formulation}\label{sec:problem}
In this article,  $\ID{\boldsymbol \cdot}$ is the indicator function,  $\rho(\boldsymbol \cdot)$ is the spectral radius of a matrix,    $\DIAG(\boldsymbol \cdot)$ is a block diagonal matrix, $\NF{\boldsymbol \cdot}$ is the Frobenius norm of a matrix, and  $\VAR(\boldsymbol \cdot)$ is the variance of a random variable.  For any $n \in \mathbb{N}$, $x_{1:n}$ is the vector $(x_1,\ldots,x_n)$ and  $\mathbb{N}_n$ is  the finite set $\{1,\ldots,n\}$. For any vectors $x,y$ and $z$,  $\VEC(x,y,z)=[x^\intercal, y^\intercal,z^\intercal]^\intercal$ and  for any matrices $A,B$ and $C$ with the same number of columns, $\ROW(A,B,C)=[A^\intercal,B^\intercal,C^\intercal ]^\intercal$. For any square matrix $A$, $A \geq 0$ and $A>0$  mean that matrix $A$  is positive semi-definite and positive definite, respectively.  Also, $\mathbf{I}$ refers to the identity matrix  and $\mathbf 0$  to a  matrix with zero arrays.
\subsection{Model}
Consider a decentralized stochastic control  system with $n \in \mathbb{N}$ decision makers (agents).   The agents are partitioned into $S  \in \mathbb{N}$ disjoint sub-populations (sub-systems) with $n(s)  \leq n$ agents, where $\sum_{s=1}^S n(s)=n$. For any sub-population $s \in \mathbb{N}_S$, let $x^i_t \in \mathbb{R}^{d^s_x}$, $u^i_t \in \mathbb{R}^{d^s_u}$ and $w^i_t \in \mathbb{R}^{d^s_x}$ denote  the state, action and noise of agent $i  \in \mathbb{N}_{n(s)}$ at time $t \in \mathbb{N}$, respectively. In addition, let $\alpha^{i,j}(s) \in \mathbb{R}$ denote the \emph{influence factor} of agent $i  \in \mathbb{N}_{n(s)}$  on the $j$-th feature of the sub-population $s \in \mathbb{N}_S$, $j \in \mathbb{N}_{f(s)}$, $f(s) \in \mathbb{N}$.   The influence factors are  orthogonal vectors in  the feature space  such that
\begin{equation}\label{eq:orthogonal_features}
\frac{1}{n(s)} \sum_{i=1}^{n(s)} \alpha^{i,j}(s) \alpha^{i,j'}(s)=\ID{j=j'}, \quad j,j' \in \mathbb{N}_{f(s)}.
\end{equation}
For any feature $j \in \mathbb{N}_{f(s)}$ of sub-population $s \in \mathbb{N}_S$, define the following linear regressions:
\begin{equation}
\bar x_t^j(s)=\frac{1}{n(s)}\sum_{i=1}^{n(s)} \alpha^{i,j}(s) x^i_t,\quad \bar u_t^j(s)=\frac{1}{n(s)}\sum_{i=1}^{n(s)} \alpha^{i,j}(s) u^i_t.
\end{equation} 
From~\cite{Jalal2019MFT,Jalal2019Automatica,Jalal2019risk,Jalal2020Nash,Jalal2020CCTA},  we refer to the above linear regressions as \emph{deep states} and \emph{deep actions}.  At any time $t \in \mathbb{N}$, define $\bar{\mathbf x}_t:=\VEC((\bar x_t^j(s))_{j=1}^{f(s)})_{s=1}^S$ and $\bar{\mathbf u}_t:=\VEC((\bar u_t^j(s))_{j=1}^{f(s)})_{s=1}^S$. Let the initial states of the agents of each sub-population $s \in \mathbb{N}_S$ be independent and identically distributed (i.i.d.)  Gaussian random vectors with positive covariance matrix  $\Sigma_x(s)$. The state evolution  of agent $i \in \mathbb{N}_{n(s)}$ in sub-population $s \in \mathbb{N}_S$  is described   by:
\begin{equation}\label{eq:dynamics_original}
\Compress
x^i_{t+1}=A(s) x^i_t+B(s) u^i_t+  \sum_{j=1}^{f(s)} \alpha^{i,j}(s)(\bar A^j(s) \bar{\mathbf x}_t + \bar  B^j(s) \bar{\mathbf u}_t )+ w^i_t,
\end{equation}
where $\{w^i_t\}_{t=1}^\infty$ is an i.i.d. zero-mean Gaussian random vector with positive covariance  matrix $\Sigma_w(s)$. The per-step cost of agent $i \in \mathbb{N}_{n(s)}$ of sub-population $s \in \mathbb{N}_S$ is defined as:
\begin{equation}\label{eq:cost_original}
c^i_t= (x^i_t)^\intercal Q(s) x^i_t+(u^i_t)^\intercal R(s) u^i_t+ \bar{\mathbf x}_t^\intercal \bar Q(s) \bar{\mathbf x}_t+ \bar{\mathbf u}_t^\intercal \bar R(s) \bar{\mathbf u}_t,
\end{equation}
where $Q(s)$, $R(s)$, $\bar Q(s)$ and $\bar R(s)$ are symmetric matrices with appropriate dimensions.  The team (social welfare) cost function at time $t \in \mathbb{N}$ is  given by:
\begin{equation}\label{eq:bar_c_per}
\bar c_t= \sum_{s=1}^S \frac{\mu(s)}{n(s)} \sum_{i=1}^{n(s)} c^i_t,
\end{equation}
where $\mu(s) >0$  determines  the importance of the cost of  agents of sub-population $s \in \mathbb{N}_S$  with respect to other sub-populations.  It is assumed that the primitive random vectors $\{ \{ \{x^i_1\}_{i=1}^{n(s)}\}_{s=1}^{S},   \{ \{w^i_1\}_{i=1}^{n(s)}\}_{s=1}^{S},  \{ \{w^i_2\}_{i=1}^{n(s)}\}_{s=1}^{S}, \ldots\}$~are defined on a common probability space, and are mutually independent across time and space. 

\begin{definition}[Weakly coupled agents~\cite{Jalal2019risk}]\label{def:weak}
\emph{The agents are said to be weakly coupled in the  dynamics if the coupling term in~\eqref{eq:dynamics_original} can be expressed as:
$\sum_{j=1}^{f(s)} \alpha^{i,j}(s)(\bar A^j(s) \bar x^j_t(s) + \bar B^j(s) \bar u^j_t(s))$.
Similarly, the agents are said to be weakly coupled in the  cost function if the coupling term in~\eqref{eq:cost_original} can be represented as:
$\sum_{s=1}^{S} \mu(s)\sum_{j=1}^{f(s)} (\bar x^j_t(s))^\intercal  \bar Q^j(s) \bar x^j_t(s)+ (\bar u^j_t(s))^\intercal  \bar R^j(s) \bar u^j_t(s)$.
Weakly coupling often arises in natural systems with equivariant structure.}
\end{definition}

The information structure considered here is called \emph{deep-state sharing} (DSS), where each agent $i \in \mathbb{N}_{n(s)}$ of sub-population $s \in \mathbb{N}_S$ observes its local state as well as the deep states, i.e., $
u^i_t=g^i_t(x^i_{1:t}, \bar{\mathbf x}_{1:t})$,
where $g^i_t$ is the control law at time $t \in \mathbb{N}$. Notice that DSS is a non-classical information structure wherein each agent has a different information set.  
\subsection{Problem statement}
 Given any \emph{risk factor}  $\lambda >0$, define the following objective function:
$J_{n,\lambda}:= \limsup_{T \rightarrow \infty} \frac{1}{\lambda T} \log \Exp{ e^{\lambda \sum_{t=1}^T \bar c_t}}$.
Note that for a small risk factor  $\lambda$, one has:
\begin{equation}\label{eq:var_relationship}
J_{n,\lambda} \approx \Exp{\limsup_{T \rightarrow \infty}\frac{1}{T}\sum_{t=1}^T  \bar c_t}+ \frac{\lambda}{2}
\VAR(\limsup_{T \rightarrow \infty}\frac{1}{T}\sum_{t=1}^T  \bar c_t).
\end{equation}
From~\eqref{eq:var_relationship}, it implies that  risk-factor $\lambda$  balances the trade off between optimality (where $\lambda \rightarrow 0$) and robustness (where robustness is defined in terms of minimum variance). To have a well-posed problem, it is assumed that all matrices defined above are uniformly bounded  in time and space, and that the set of admissible  actions are square integrable.  Let $\mathbf g=:\{\{\{ g^i_t\}_{i=1}^{n(s)}\}_{s=1}^S\}_{t=1}^\infty$ denote the strategy of all agents.
\begin{problem}\label{prob1}
Develop  model-based  gradient descent and natural policy gradient descent algorithms to compute  the optimal risk-sensitive strategy $\mathbf g^\ast$ such that for any  strategy $\mathbf g$, the following inequality holds:
$J_{n,\lambda}(\mathbf g^\ast) \leq J_{n,\lambda}(\mathbf g)$.
\end{problem}

\begin{problem}\label{prob2}
Develop   model-free gradient descent and natural policy gradient descent algorithms to learn  the optimal risk-neutral strategies  $\mathbf g^\ast$, i.e., when $\lambda \rightarrow 0$.
\end{problem}

\begin{remark}
 \emph{For  the special case of single sub-population (i.e. $S=1$), single agent (i.e. $n(s)=1$, $s \in \mathbb{N}_S$) and single feature (i.e. $f(s)=1$, $s\in \mathbb{N}_S$),  deep structured teams   reduce to the classical   single-agent control problems~\cite{jacobson1973optimal,whittle1981risk,bacsar2008h}.}
\end{remark}
%
\section{Main Results for Problem~\ref{prob1}}\label{sec:main}
In this section, we first present the solution of Problem~1  in terms of Riccati equations. Then, we  establish the global convergence of model-based  policy gradient algorithms.
 From~\cite{Jalal2019risk},  we define a gauge transformation for any agent $i \in \mathbb{N}_{n(s)}$ in sub-population $s \in \mathbb{N}_S$ at time $t \in \mathbb{N}$:
\begin{align}\label{eq:Gauge}
\begin{cases}
\Delta x^i_t:=x^i_t -\sum_{j=1}^{f(s)} \alpha^{i,j}(s) \bar x^j_{t}(s),\\
\Delta u^i_t:=u^i_t -\sum_{j=1}^{f(s)} \alpha^{i,j}(s) \bar u^j_{t}(s),\\
\Delta w^i_t:=w^i_t -\sum_{j=1}^{f(s)} \alpha^{i,j}(s) \bar w^j_{t}(s), 
\end{cases}
\end{align}
where $ \bar w^j_t(s):=\frac{1}{n(s)}\sum_{i=1}^{n(s)} \alpha^{i,j}(s) w^i_t$, $\forall j \in \mathbb{N}_{f(s)}$.  The gauge transformation induces the following linear dependences:
\begin{equation}\label{eq:linear_dependence}
\begin{cases}
\sum_{i=1}^{n(s)} \sum_{j=1}^{f(s)} \alpha^{i,j}(s) \Delta x^i_t=\mathbf 0,\\ \sum_{i=1}^{n(s)} \sum_{j=1}^{f(s)} \alpha^{i,j}(s) \Delta u^i_t=\mathbf 0,\\ \sum_{i=1}^{n(s)} \sum_{j=1}^{f(s)} \alpha^{i,j}(s) \Delta w^i_t=\mathbf 0.
\end{cases}
\end{equation}
Subsequently, the dynamics of the  $j$-th deep state of sub-population~$s \in \mathbb{N}_S$ can be represented as follows:
\begin{equation}
\Compress
\bar x^j_{t+1}(s)=A(s) \bar x^j_t(s) + B(s) \bar u^j_t(s) +\bar A^j(s) \bar{ \mathbf x}_t + \bar B^j(s) \bar{\mathbf u}_t+ \bar w^j_t(s).  
\end{equation}

For any $s \in \mathbb{N}_S$ and  $t \in \mathbb{N}$, define the following matrices:
\begin{equation}
\begin{cases}
\tilde {\mathbf A}_t(s):=\DIAG(A_t(s))_{f(s)},\quad 
\tilde {\mathbf B}_t(s):=\DIAG(B_t(s))_{f(s)},\\
\bar {\mathbf A}_t(s):=[\mathbf 0_{f(s)d^s_x\times f(1)d^1_x},\ldots,\tilde{\mathbf A}_t(s),\ldots,\mathbf 0_{f(s)d^s_x\times f(S)d^S_x}]\\
\quad +\ROW(\bar A^1_t(s),\ldots,\bar A^{f(s)}_t(s)),\\
\bar {\mathbf B}_t(s):=[\mathbf 0_{f(s)d^s_u\times f(1)d^1_u},\ldots,\tilde{\mathbf B}_t(s),\ldots,\mathbf 0_{f(s)d^s_u\times f(S)d^S_u}]\\
\quad +\ROW(\bar B^1_t(s),\ldots,\bar B^{f(s)}_t(s)),\\
\bar{\mathbf A}_t  :=\ROW(\bar{\mathbf A}_t(1),\ldots, \bar{\mathbf A}_t(S)), \\
\bar{\mathbf B}_t:=\ROW(\bar{\mathbf B}_t(1),\ldots, \bar{\mathbf B}_t(S)).
\end{cases}
\end{equation}
One can then write:
$\bar{\mathbf x}_{t+1}=\bar{\mathbf A} \bar{\mathbf x}_t+\bar{\mathbf B} \bar{\mathbf u}_t+ \bar{\mathbf w}_t$,
where  $\bar{\mathbf w}_t:=\VEC((\bar w_t^j(s))_{j=1}^{f(s)})_{s=1}^S$. From~\eqref{eq:dynamics_original} and~\eqref{eq:Gauge}, it  follows that for any $i \in \mathbb{N}_{n(s)}$ and $s \in \mathbb{N}_S$:
$\Delta x^i_{t+1}=A(s) \Delta x^i_t+ B(s) \Delta u^i_t+\Delta w^i_t$.
In addition, one has the following orthogonal relations:
\begin{equation}\label{eq:orthogonal_relation}
\begin{cases}
\sum_{i=1}^{n(s)} \sum_{j=1}^{f(s)} \alpha^{i,j}(s) (\Delta x^i_t)^\intercal Q(s) \bar x^j_t(s)= 0,\\
\sum_{i=1}^{n(s)} \sum_{j=1}^{f(s)} \alpha^{i,j}(s) (\Delta u^i_t)^\intercal R(s) \bar u^j_t(s)= 0.\\
\end{cases}
\end{equation}
Define the following matrices $\bar {\mathbf Q}_t(s):=\DIAG(Q_t(s))_{f(s)}$, $\bar {\mathbf R}_t(s):= \DIAG(R_t(s))_{f(s)}$, $s \in \mathbb{N}_S$, $t \in \mathbb{N}$, and
\begin{align}
\bar{\mathbf Q}:=\DIAG\big( \mu(s)\DIAG(\underbrace{Q(s),\ldots, Q(s)}_{f(s)}) \big)_{s=1}^S+\sum_{s=1}^S \mu(s) \bar Q(s),\\
\bar{\mathbf R}:=\DIAG\big( \mu(s)\DIAG(\underbrace{R(s),\ldots, R(s)}_{f(s)}) \big)_{s=1}^S+\sum_{s=1}^S \mu(s) \bar R(s).
\end{align}
The cost function~\eqref{eq:bar_c_per} can be reformulated as:
$
\bar c_t=\bar{\mathbf x}_t^\intercal \bar{\mathbf Q} \bar{\mathbf x}_t + \bar{\mathbf u}_t^\intercal \bar{\mathbf R} \bar{\mathbf u}_t +\sum_{s=1}^S \frac{\mu(s)}{n(s)} \sum_{i=1}^{n(s)} (\Delta x^i_t)^\intercal Q(s) \Delta x^i_t + (\Delta u^i_t)^\intercal R(s) \Delta u^i_t. 
$
 Define also the following deep Riccati equation (that consists of $S+1$  decoupled Riccati equations):
\begin{equation}\label{eq:DR_1}
\begin{cases}
P_n(s)=Q(s)+   A^\intercal(s) \tilde P_n(s) A(s) -  A^\intercal(s) \tilde P_n(s) B(s)\\
\qquad \times ( R(s) + B^\intercal(s) \tilde P_n(s) B(s))^{-1} B^\intercal(s) \tilde P_n(s) A(s),\\
\tilde P_n(s)=P_n(s)(\mathbf I_{d^s_x \times d^s_x} - 2 \lambda \frac{\mu(s)}{n(s)} \Sigma_w(s) P_n(s))^{-1},
\end{cases}
\end{equation}
\begin{equation}\label{eq:DR_2}
\begin{cases}
\mathbf P_n=\bar{\mathbf Q}+  \bar{\mathbf A}^\intercal \tilde{\mathbf P}_n \bar{\mathbf A} -  \bar{\mathbf A}^\intercal \tilde{\mathbf  P}_n \bar{\mathbf B} ( \bar{\mathbf R} + \bar{\mathbf B}^\intercal \tilde{\mathbf  P}_n \bar{\mathbf B})^{-1} \bar{\mathbf B}^\intercal \tilde{\mathbf  P}_n \bar{\mathbf A},\\
\tilde{\mathbf P}_n=\mathbf P_n(\mathbf I - 2 \lambda \boldsymbol \Sigma_w \mathbf P_n)^{-1},\\
\boldsymbol \Sigma_w:=\DIAG\big( \frac{1}{n(s)} \DIAG(\underbrace{\Sigma_w(s),\ldots,\Sigma_w(s)}_{f(s)})  \big)_{s=1}^S.
\end{cases}
\end{equation}
The dimensions of the above Riccati equations are independent of  the number of agents in each sub-population~$s \in \mathbb{N}_S$.
\begin{remark}
\emph{For the weakly-coupled case in Definition~\ref{def:weak}, the  deep Riccati equation decomposes further into $S+\sum_{s=1}^S f(s)$ smaller Riccati equations, where~\eqref{eq:DR_2}  can be expressed for any $j \in \mathbb{N}_{f(s)}$ and any $s \in \mathcal{S}$ as follows:
\begin{equation}\label{eq:DR_3}
\begin{cases}
\bar P^j_n(s)=Q(s)+  \bar Q^j(s) + (A(s)+\bar A^j(s))^\intercal \tilde{\bar P}^j_n(s)\\
 \times (A(s)+\bar A^j(s)) -  (A(s) + \bar A^j(s))^\intercal \tilde{\bar P}^j_n (B(s)+\bar B^j(s))\\
\times  \big( ( R(s)+\bar R^j(s)) 
 + (B(s)+\bar B^j(s))^\intercal  \tilde{\bar P}^j_n(s) \\
 \times (B(s)\hspace{-.1cm}+\hspace{-.1cm}\bar B^j(s))\big)^{\hspace{-.1cm}-1} 
 (B(s)\hspace{-.1cm}+\hspace{-.1cm}\bar B^j(s))^\intercal \tilde{\bar P}^j_n(s) (A(s)+\bar A^j(s)),\\
\tilde{\bar P}^j_n(s)=\bar P^j_n(s)(\mathbf I_{d^s_x \times d^s_x} - 2 \lambda \frac{\mu(s)}{n(s)} \Sigma_w(s) \bar P^j_n(s))^{-1},
\end{cases}
\end{equation}
where $
\mathbf P_n= \DIAG\big(\mu(s)\DIAG( \bar P^j_n(s) )_{j=1}^{f(s)}\big)_{s=1}^S$ and $\tilde{\mathbf P}_n= \DIAG\big(\mu(s)\DIAG( \tilde{\bar P}^j_n(s) )_{j=1}^{f(s)}\big)_{s=1}^S$.
}
\end{remark}

\begin{assumption}\label{ass:unique}
The followings hold:
\begin{itemize}
\item [(I)] $Q(s)\geq 0$, $\bar{\mathbf Q} \geq 0$, $R(s) >0$ and $\bar{\mathbf R}>0$, $\forall s \in \mathbb{N}_S$. 
\item[(II)] Pairs $(A(s),B(s))$, $s \in \mathbb{N}_S$, and  $(\bar{\mathbf A},\bar{\mathbf B})$ are stablizable (controllable). In addition, pairs $(A(s), Q^{1/2}(s))$,  $s \in \mathbb{N}_S$, and $(\bar{\mathbf A}, \bar{ \mathbf Q}^{1/2})$ are detectable (observable).
\item[(III)] Equations~\eqref{eq:DR_1} and~\eqref{eq:DR_2} admit positive solutions $P_n(s)>0$, $\forall s \in \mathbb{N}_S$, and $\mathbf P_n>0$. In addition, $\mathbf I_{d^s_x \times d^s_x}- 2\lambda \frac{\mu(s)}{n(s)}\Sigma_w(s) P_n(s) >0$, $\forall s \in \mathbb{N}_S$, and $\mathbf I_{\sum_{s=1}^S f(s) d^s_x \times\sum_{s=1}^S f(s) d^s_x }- 2\lambda \boldsymbol \Sigma_w \mathbf P_n >0$.
\end{itemize}
\end{assumption}
\begin{remark}
\emph{Part (I) of Assumption~\ref{ass:unique} is a standard convexity condition and  Part (II) is required  to   ensure that the  system is stablizable. Part (III) is a standard condition in risk-sensitive  LQ problems that guarantees  the deep  Riccati equation, presented in~\eqref{eq:DR_1} and~\eqref{eq:DR_2}, admit a unique positive definite solution. Suppose  matrices in the dynamics~\eqref{eq:dynamics_original} and cost functions~\eqref{eq:cost_original} are independent of the size of sub-populations $n(s)$, $s \in \mathbb{N}_S$; then,  if the risk-factor $\lambda$  decreases and/or the number of agents (i.e. $(n(s)$, $s \in \mathbb{N}_S$) increases, the  positiveness condition in Part (III) gets more relaxed such that  it automatically holds if $\lambda=0$ and/or $n(s)=\infty$.}
\end{remark}
Since the  certainty equivalence theorem does  not hold in the risk-sensitive case, which is in  contrast to the risk-neutral model,  we present a few key covariance properties.
\begin{lemma}\label{lemma:correlation}
For any $s \in \mathbb{N}_S$, the followings  hold for any $i \neq k \in \mathbb{N}_{n(s)}$, $j \neq m \in \mathbb{N}_{f(s)}$ and $t \in \mathbb{N}$: $\Exp{\Delta w^i_t (\Delta w^i_t)^\intercal}=\big(1- \frac{1}{n(s)} \sum_{j=1}^{f(s)} (\alpha^{i,j}(s))^2\big)\Sigma_w(s), 
\Exp{\Delta w^i_t (\Delta w^k_t)^\intercal}=$ $-\frac{1}{n(s)} \sum_{j=1}^{f(s)} \alpha^{i,j}(s) \alpha^{k,j} \Sigma_w(s)$, $\Exp{\bar w^j_t(s) (\bar w^j_t(s))^\intercal}=\frac{1}{n(s)}\Sigma_w$, $\Exp{\Delta w^i_t (\bar w^j_t(s))^\intercal}=\mathbf 0$, $
\Exp{\bar w^j_t(s) (\bar w^m_t(s))^\intercal}=\mathbf 0$.
\end{lemma}
\begin{proof}
The proof follows directly from~\eqref{eq:orthogonal_features} and~\eqref{eq:Gauge}. 
\end{proof}
From~\eqref{eq:DR_1} and~\eqref{eq:DR_2}, define for any sub-population  $s \in \mathbb{N}_S$: $
\theta^\ast_n(s):=-( R(s) +B^\intercal(s) \tilde P_n(s) B(s))^{-1} B^\intercal(s) \tilde P_n(s) A(s)$ and $
\bar{\boldsymbol \theta}^\ast_n=:-(\bar{\mathbf R} +\bar{\mathbf B}^\intercal \tilde{\mathbf P}_n \bar{\mathbf B})^{-1} \bar{\mathbf B}^\intercal \tilde{\mathbf P}_n \bar{\mathbf A}$.
Let  $\bar \theta^{\ast,j}_n(s)$ denote a  block of matrix $\bar{\boldsymbol \theta}^\ast_n$  that is associated with the $j$-th feature of sub-population $s \in \mathbb{N}_S$, $ j \in \mathbb{N}_{f(s)}$.
\begin{theorem}[Model-known solution~\cite{Jalal2019risk}]\label{thm:model_known_optimal}
Let Assumption~\ref{ass:unique} hold. There exists a unique stationary optimal linear strategy such that for any $i \in \mathbb{N}_{n(s)}$ and $s \in \mathbb{N}_S$ at time $t \in \mathbb{N}$:
\begin{equation}\label{eq:optimal_model_known}
u^i_t=\theta^\ast_n(s) x^i_t - \sum_{j=1}^{f(s)} \alpha^{i,j}(s) \theta^\ast_n(s) \bar x^j_t(s)+ \sum_{j=1}^{f(s)} \alpha^{i,j}(s) \bar \theta^{\ast,j}_n(s) \bar{\mathbf x}_t.
\end{equation}
\end{theorem}
\begin{proof}
The proof follows from the  linearly dependent equations~\eqref{eq:linear_dependence}, orthogonal relations~\eqref{eq:orthogonal_relation} and  covariance  properties in Lemma~\ref{lemma:correlation}, leading to a low-dimensional representation of the  solution. For more details, see~\cite[Theorem 1]{Jalal2019risk}.
\end{proof}
At the initial time,  every agent at any  sub-population $s \in \mathbb{N}_S$ solves two Riccati equations: one Riccati equation in~\eqref{eq:DR_1}  with the dimension $d^s_x \times d^s_x$ (assigned  specifically to the sub-population $s$) and one common Riccati equation in~\eqref{eq:DR_2} with the dimension $\sum_{s=1}^S f(s) d^s_x \times \sum_{s=1}^S f(s) d^s_x$.  For the special case of weakly coupled agents, the common Riccati equation decomposes into $\sum_{s=1}^S f(s)$ smaller Riccati equations. In this case,   every agent needs to solve only $f(s)+1$ Riccati equations in~\eqref{eq:DR_1} and~\eqref{eq:DR_3} with the dimensions $d^s_x \times d^s_x$. During the control process, each agent computes its action according to~\eqref{eq:optimal_model_known} based on the above Riccati solutions, its local (private)  state and  influence factors as well as common (public) deep states.

\subsection{Model-based approach}
From Theorem~\ref{thm:model_known_optimal},  the optimization problem in action space is strictly convex and there is no loss of optimality in restricting attention to stationary linear strategies of the form~\eqref{eq:optimal_model_known}. However, the convexity in  action space does not lead to  the convexity  in  policy space; see a simple counterexample in~\cite{fazel2018global}.  In what follows, we provide an analytical proof showing that  policy gradient methods converge to the globally optimal solution~\eqref{eq:optimal_model_known} based on the concepts of gradient domination and locally Lipschitz continuity in~\cite{fazel2018global}. Consider a  stationary strategy $\boldsymbol \theta:=\{ \theta(1),\ldots,\theta(S), \bar{\boldsymbol \theta} \}$, where $\theta(s)$ is a $d^s_u \times d^s_x$ matrix, $s \in \mathbb{N}_S$, and $\bar{\boldsymbol \theta}$ is a  $\sum_{s=1}^S f(s) d^u_s \times \sum_{s=1}^S f(s) d^s_x$ matrix.  At any time $t \in \mathbb{N}$, one has:
$\Delta u^i_t= \theta(s) \Delta x^i_t, \hspace{.1cm} i \in \mathbb{N}_{n(s)}, s \in \mathbb{N}_S$  and $
\bar{\mathbf u}_t= \bar{\boldsymbol \theta} \bar{\mathbf x}_t$.
For any sub-population $s \in \mathbb{N}_S$, define:
\begin{equation}
\begin{cases}
P_{\theta(s)}=Q(s)+A^\intercal(s) P_{\theta(s)} A(s)\\
\qquad + (A(s) -B(s)\theta(s))^\intercal \tilde P_{\theta(s)}(A(s) -B(s)\theta(s)),\\
\tilde P_{\theta(s)}= P_{\theta(s)} ( \mathbf I_{d^s_x\times d^s_x} - 2\lambda \frac{\mu(s)}{n(s)} \Sigma_w(s) P_{\theta(s)})^{-1},\\
\mathbf P_{\bar{\boldsymbol \theta}}=\bar{\mathbf Q}+\bar{\mathbf A}^\intercal \mathbf P_{\bar{\boldsymbol \theta}} \bar{\mathbf A}+  (\bar{\mathbf A} -\bar{\mathbf B} \bar{\boldsymbol \theta})^\intercal \tilde{\mathbf  P}_{\bar{\boldsymbol \theta}}(\bar{\mathbf A} -\bar{\mathbf B} \bar{\boldsymbol \theta}),\\
\tilde{\mathbf  P}_{\bar{\boldsymbol \theta}}= \bar{\mathbf P}_{\bar{\boldsymbol \theta}} ( \mathbf I- 2\lambda  \boldsymbol \Sigma_w \mathbf P_{\bar{\boldsymbol \theta}})^{-1}.
\end{cases}
\end{equation}

Following~\cite[Lemma 3.3]{zhang2019policyrisk}, we take the gradient of the cost function with respect to $\boldsymbol \theta$ and obtain
$
\nabla_{\boldsymbol \theta} J(\boldsymbol \theta):=\{\nabla_{\theta(1)} J(\boldsymbol \theta),\ldots, \nabla_{\theta(S)} J(\boldsymbol \theta),\nabla_{\bar{\boldsymbol \theta}} J(\boldsymbol \theta)  \}$,
where\footnote{For simplicity, it is assumed in~\cite{zhang2019policyrisk} that  the initial states have zero mean.} 
\begin{equation}\label{eq:GD_1}
\begin{cases}
\nabla_{\theta(s)} J(\boldsymbol \theta)=   2E_{\theta(s)}  \Sigma_{\theta(s)},\\
 \hspace{-.1cm}E_{\theta(s)}  \hspace{-.1cm}:=  \hspace{-.1cm} (R(s)\hspace{-.1cm}+\hspace{-.1cm} B^\intercal(s) \tilde P_{\theta(s)}B(s)) \theta(s) \hspace{-.1cm}- \hspace{-.1cm}    B^\intercal \tilde P_{\theta(s)} A(s),\\
\Sigma_{\theta(s)}:=\sum_{t=1}^\infty  \big([\mathbf I-2\lambda \frac{\mu(s)}{n(s)} \Sigma_w(s) P_{\theta(s)}]^{-\intercal}\\
\times (A(s) - B(s) \theta(s)) \big)^t [ \mathbf I-2\lambda \frac{\mu(s)}{n(s)} \Sigma_w(s) P_{\theta(s)}]^{-1}  \frac{\mu(s)}{n(s)} \Sigma_w(s)\\
\times  \big( (A(s) - B(s) \theta(s))^\intercal  [ \mathbf I-2\lambda \frac{\mu(s)}{n(s)} \Sigma_w(s) P_{\theta(s)}]\big)^t,\\
\nabla_{\bar{\boldsymbol \theta}} J(\boldsymbol \theta)= 2  \mathbf E_{\bar{\boldsymbol \theta}} \boldsymbol \Sigma_{\bar{\boldsymbol \theta}},\\
\mathbf E_{\bar{\boldsymbol \theta}}:= (\bar{\mathbf R}+ \bar{\mathbf B}^\intercal \tilde{\mathbf P}_{\bar{\boldsymbol \theta}}\bar{\mathbf B}\Big) \bar{\boldsymbol \theta}- \bar{\mathbf B}^\intercal \tilde{\mathbf P}_{\bar{\boldsymbol \theta}} \bar{\mathbf A},\\
\boldsymbol \Sigma_{\bar{\boldsymbol \theta}}:=\sum_{t=1}^\infty  \big([ \mathbf I-2\lambda \boldsymbol{ \Sigma}_w \bar{\mathbf P}_{\bar{\boldsymbol \theta}}]^{-\intercal}(\bar{\mathbf A} - \bar{\mathbf B} \bar{\boldsymbol \theta}) \big)^t \\
\times [ \mathbf I-2\lambda \boldsymbol{ \Sigma}_w  \bar{\mathbf P}_{\bar{\boldsymbol \theta}}]^{-1} \boldsymbol{\Sigma}_w \big( (\bar{\mathbf A} - \bar{\mathbf B} \bar{\boldsymbol \theta})^\intercal  [ \mathbf I-2\lambda \boldsymbol{\Sigma}_w \bar{\mathbf P}_{\bar{\boldsymbol \theta}}]\big)^t.
\end{cases}
\end{equation}
For the special case of risk-neutral $\lambda \rightarrow 0$ and/or $n(s) \rightarrow \infty$, one arrives at
\begin{equation}\label{eq:risk_neutral_pa}
\begin{cases}
E_{\theta(s)} :=   (R(s)+B^\intercal(s)  P_{\theta(s)}B(s)) \theta(s) -    B^\intercal  P_{\theta(s)} A(s),\\
\Sigma_{\theta(s)}= \frac{\mu(s)}{n(s)} \sum_{i=1}^{n(s)} \sum_{t=1}^\infty  \Delta x^i_t (\Delta x^i_t)^\intercal,\\
\mathbf E_{\bar{\boldsymbol \theta}}= (\bar{\mathbf R}+ \bar{\mathbf B}^\intercal \mathbf P_{\bar{\boldsymbol \theta}}\bar{\mathbf B}\Big) \bar{\boldsymbol \theta}- \bar{\mathbf B}^\intercal \mathbf P_{\bar{\boldsymbol \theta}} \bar{\mathbf A},\\
\boldsymbol \Sigma_{\bar{\boldsymbol \theta}}=\sum_{t=1}^\infty  \bar{\mathbf  x}_t (\bar{\mathbf x}_t)^\intercal.
\end{cases}
\end{equation}
We propose two gradient methods  described below, where $k \in \mathbb{N}$ denotes the iteration.
\begin{itemize}
\item \textbf{Policy gradient descent}: 
\begin{equation}\label{eq:GD}
\begin{cases}
 \theta_{k+1}(s)= \theta_k(s) - \eta \nabla_{\theta_k(s)} J(\boldsymbol \theta ),\quad s \in \mathbb{N}_S,\\
 \bar{\boldsymbol \theta}_{k+1}= \bar{\boldsymbol \theta}_{k} - \eta \nabla_{\bar{\boldsymbol \theta}} J(\boldsymbol \theta).
\end{cases}
\end{equation}
\item \textbf{Natural policy gradient descent}:
\begin{equation}\label{eq:NPGD}
\begin{cases}
 \theta_{k+1}(s)= \theta_k(s) - \eta \nabla_{\theta_k(s)} J(\boldsymbol \theta ) \Sigma_{\theta(s)}^{-1},\quad s \in \mathbb{N}_S,\\
 \bar{\boldsymbol \theta}_{k+1}= \bar{\boldsymbol \theta}_{k} - \eta \nabla_{\bar{\boldsymbol \theta}} J(\boldsymbol \theta) \boldsymbol \Sigma_{\bar{\boldsymbol \theta}}^{-1}.
\end{cases}
\end{equation}
\end{itemize}
We now make an assumption  that the initial policy is stable, which  is a standard assumption.
\begin{assumption}\label{ass:stable}
For the initial policy, $\rho(A(s) -B(s) \theta_1(s))<1$, $s \in \mathbb{N}_S$, and  $\rho(\bar{\mathbf A} -\bar{\mathbf B} \bar{\boldsymbol  \theta}_1)<1$. In addition, part III of Assumption~\ref{ass:unique} holds for $P_{\theta_1(s)} >0$, $s \in \mathbb{N}_S$,  and $\mathbf{P}_{\bar{\boldsymbol \theta}_1} >0$.
\end{assumption}

 The solution of Riccati equation in the risk-sensitive model has a complex relationship with the policy, making it difficult to establish the  gradient dominance. Fortunately, when the number of agents goes to infinity,  the solution of the risk-sensitive Riccati equation converges to that of the risk-neutral one. This observation enables us to  establish an asymptotic global convergence result for the risk-sensitive cost function.

\begin{theorem}\label{thm2}
Let Assumptions~\ref{ass:unique} and~\ref{ass:stable} hold.  For a sufficiently small risk factor $\lambda$ and/or sufficiently large population $n(s)$, $\forall s \in \mathbb{N}_S$, the policy gradient algorithms in~\eqref{eq:GD} and~\eqref{eq:NPGD}  converge to the globally optimal solution $\boldsymbol \theta^\ast$ for an  adequately small step size $\eta$.
\end{theorem}
\begin{proof}
 Let $\boldsymbol \theta$  have a finite cost. 	Following from~ \cite[Lemma 11]{fazel2018global},  we obtain an upper bound on  the distance  between the cost function and   its  optimal value  in terms of   $E_{\theta(s)}$ and $\Sigma_{\theta(s)}$, $s \in \mathbb{N}_S$, as well as $ \mathbf E_{\bar{\boldsymbol \theta}}$ and $\boldsymbol \Sigma_{\bar{\boldsymbol \theta}}$ (which  represent the gradients in~\eqref{eq:GD_1}). Therefore,   there exists  a positive constant $L_1 (\boldsymbol \theta^\ast)$ such that $|J(\boldsymbol \theta)- J(\boldsymbol \theta^*)| \leq  L_1 (\boldsymbol \theta^\ast) \NF{\nabla_{\boldsymbol \theta} J(\boldsymbol \theta)}^2$. This inequality is known as gradient domination (or PL inequality~\cite{polyak1963gradient}). Furthermore,  we proceed according to~\cite[Lemmas 15 and 16]{malik2020derivative} to show that the cost  and gradient are  locally Lipschitz functions in the neighbourhood of~$\boldsymbol \theta$, where the Lipschitz constants depend on  $\varepsilon(\boldsymbol \theta)$.  In particular, given any $\boldsymbol \theta '$ satisfying the inequality $\NF{\boldsymbol \theta' - \boldsymbol \theta} < \varepsilon(\boldsymbol \theta)$, there exist  positive constants $L_2(\boldsymbol \theta)$ and $L_3(\boldsymbol \theta)$ such that  
$
	|J(\boldsymbol \theta ')- J(\boldsymbol \theta )|\leq L_2(\boldsymbol \theta )\NF{\boldsymbol \theta'-\boldsymbol \theta}$ 
	and
		$
	\NF{\nabla J_{\boldsymbol \theta'}(\boldsymbol \theta')- \nabla J_{\boldsymbol \theta}(\boldsymbol \theta)}\leq L_3(\boldsymbol \theta)\NF{\boldsymbol \theta'-\boldsymbol \theta}.
$ Following the proof technique proposed in~\cite[Theorem 7]{fazel2018global} and~\cite[Theorem 2]{Masoud2020CDC}, we select a  sufficiently  small  step size $\eta$ such  that the value of the cost decreases at each iteration. In particular,  for the natural policy gradient descent and  a sufficiently  large number of iterations $K$, one has: 
	$J ({\boldsymbol \theta}_{K+1})-J (\boldsymbol \theta^*)\leq
	\big(1-\eta \eta^{-1}_{max}\big)(J ({\boldsymbol \theta}_K)-J (\boldsymbol \theta^*))$. The above recursion is contractive  for a sufficiently small step size $\eta \leq \eta_{max}$.
\end{proof}

\section{Main results for Problem~\ref{prob2}}\label{sec:main2}

In this subsection, we propose  model-free policy gradient descent and natural policy gradient algorithms for the special case of risk neutral problem (where $\lambda \rightarrow 0$) such that
$
J(\boldsymbol \theta)= \limsup_{T \rightarrow \infty} \frac{1}{T}\Exp{ \sum_{t=1}^T \bar c_t}.
$

\begin{lemma}[Finite-horizon approximation]\label{lemma:rollout}
Let  ${\tilde J_T}(\boldsymbol \theta):=\mathbb{E}[\sum_{t=1}^{T} \bar c_t]$ for any $\boldsymbol \theta$ with a finite cost function. Let also $\varepsilon(T)$ be a positive  function such that  $\lim_{T \rightarrow \infty} \varepsilon(T)=0$. Then, there exists a sufficiently large horizon~$T$  for which $|\tilde J_T(\boldsymbol \theta) - J(\boldsymbol \theta)| \leq \varepsilon(T)$.
\end{lemma}
\begin{proof}
   	The proof follows from~\cite[Lemma 26]{fazel2018global}. 
\end{proof} 
Denote by  $\mathbb{S}_r$ a  set of uniformly distributed points with norm $r>0$ and  by  $\mathbb{B}_r$ the set of all uniformly distributed points  whose norms are at most $r$. Thus, 
$
J(\boldsymbol \theta)=\mathbb{E}_{\tilde{\boldsymbol \theta} \sim \mathbb{B}_r}[J(\boldsymbol \theta+\tilde{\boldsymbol \theta})].
$  Let  $\tilde{\boldsymbol \theta}:=\{\tilde \theta(1),\ldots,\tilde \theta(S),\tilde{\bar{\boldsymbol \theta}}\}$  be a set of  independent random  matrices  whose Frobenius norm is $r$.
\begin{lemma}[Zeroth-order optimization]\label{lemma:smooth}
	For  a smoothing factor $r >0$, 
	$
	\nabla_{\theta(s)} J(\boldsymbol \theta)
= \frac{d^s_xd^s_u}{r^2} \mathbb{E}_{\tilde{\boldsymbol \theta} \sim \mathbb{S}_r }[J(\boldsymbol \theta+\tilde{\boldsymbol \theta})\tilde \theta(s)]$, $s \in \mathbb{N}_S$,  and $
	\nabla_{\bar{\boldsymbol \theta}} J(\boldsymbol \theta)
= \frac{(\sum_{s=1}^S f(s) d^s_x)(\sum_{s=1}^S f(s) d^s_u)}{r^2} \mathbb{E}_{\tilde{\boldsymbol \theta} \sim \mathbb{S}_r }[J(\boldsymbol \theta+\tilde{\boldsymbol \theta}) \tilde{\bar{\boldsymbol  \theta}}]$.
\end{lemma}
\begin{proof}
	The proof follows directly from the zeroth-order optimization approach~\cite[Lemma 1]{flaxman2004online} and the fact that  the cost function gets decoupled into $S+1$ additive terms. 
\end{proof} 

\begin{lemma}\label{lemma:ber}
Given  any $s \in \mathbb{N}_S$, let $\tilde { \theta}_1(s), \ldots, \tilde { \theta}_L(s)$, $L \in \mathbb{N}$, be   i.i.d. samples drawn uniformly  from $\mathbb{S}_r$. For any  $\varepsilon(L)>0$, the  following average converges to an $\varepsilon(L)$-neighbourhood  of the gradient  $\nabla J_{\theta(s)}(\boldsymbol \theta)$ in the Frobenius norm  with a probability greater than $1-({\frac{d^s_x d^s_u}{\varepsilon(L)}})^{-d^s_x d^s_u}$: 
$\hat \nabla_{ \theta(s)} J( \boldsymbol \theta)
	:=\frac{1}{L} \sum_{l=1}^{L}\frac{d^s_x d^s_u}{r^2}J(\boldsymbol \theta+\tilde{\boldsymbol \theta}_l)\tilde{\theta}_l(s)$.
Similarly,  let $\tilde {\bar{\boldsymbol \theta}}_1, \ldots, \tilde{\bar{\boldsymbol  \theta}}_L$, $L \in \mathbb{N}$, be   i.i.d. samples drawn uniformly  from $\mathbb{S}_r$.  Then, the  following average converges to an $\varepsilon(L)$-neighbourhood  of the gradient  $\nabla J_{\bar{\boldsymbol \theta}}(\boldsymbol \theta)$ in the Frobenius norm with a probability greater than $1-({\frac{  (\sum_{s=1}^S f(s) d^s_x)(\sum_{s=1}^S f(s) d^s_u)}{\varepsilon(L)}})^{-(\sum_{s=1}^S f(s) d^s_x)(\sum_{s=1}^S f(s) d^s_u)}$: 
		$
\hat \nabla_{\bar{\boldsymbol \theta}} J( \boldsymbol \theta)
	:=\frac{1}{L} \sum_{l=1}^{L}\frac{(\sum_{s=1}^S f(s) d^s_x)(\sum_{s=1}^S f(s) d^s_u)}{r^2}J(\boldsymbol \theta+\tilde{\boldsymbol \theta}_l)\tilde{\bar{\boldsymbol \theta}}_l$.
	\end{lemma}
	\begin{proof}
Following~the steps proposed  in~\cite[Lemma 30]{fazel2018global},  	 the proof follows  from Lemma~\ref{lemma:smooth} and Bernstein's inequality.
	\end{proof}
We now  compute an  empirical gradient for a sufficiently large number of  samples $L$ and  rollouts $T$:
	\begin{equation}\label{eq:final_1}
	\begin{cases}
\tilde  \nabla^{L,T}_{ \theta(s)} J( \boldsymbol \theta)
	:=\frac{1}{L} \sum_{l=1}^{L}\frac{d^s_x d^s_u}{r^2} \sum_{t=1}^T \tilde J_T(\boldsymbol \theta+\tilde{\boldsymbol \theta}_l)\tilde{\theta}_l(s),\\
\tilde \nabla^{L,T}_{\bar{\boldsymbol \theta}} J( \boldsymbol \theta)
	:=\frac{1}{L} \sum_{l=1}^{L}\frac{(\sum_{s=1}^S f(s) d^s_x)(\sum_{s=1}^S f(s) d^s_u)}{r^2}\\
	\hspace{2cm} \times \sum_{t=1}^T \tilde J_T(\boldsymbol \theta+\tilde{\boldsymbol \theta}_l)\tilde{\bar{\boldsymbol \theta}}_l.
	\end{cases}
	\end{equation}
\begin{theorem}
Let Assumptions~\ref{ass:unique} and~\ref{ass:stable} hold.  There exists a sufficiently small step size $\eta$ such that  the following inequality holds with a probability converging to one  as the number of samples $L$ and rollouts $T$ tend to infinity,
	$|\tilde J_{L,T}(\boldsymbol \theta)-J(\boldsymbol \theta^\ast)| \leq \varepsilon(L,T)$, 
	where $\varepsilon(L,T)=\text{poly}(1/L,1/T)$.
\end{theorem}
\begin{proof}
It results from~\cite[Theorem 31]{fazel2018global} and \cite[Theorem 3]{Masoud2020CDC} that  the following inequality  at iteration $K\in \mathbb{N}$  holds for a sufficiently small step size $\eta \leq \eta_{max}$:  
$J (\boldsymbol \theta_{K+1})-J (\boldsymbol \theta^*)\leq (1-\eta\eta_{max}^{-1})(J (\boldsymbol \theta_K)-J (\boldsymbol \theta^*))$.
At iteration $K$,  let  $\tilde{\nabla}_K$  denote  the empirical  gradient in~\eqref{eq:final_1} and $\hat{\boldsymbol{\theta}}_{K+1}=\boldsymbol{\theta}_K-\eta \tilde{\nabla}_K$ denote the iterate with  the empirical  gradient.  Due to the locally Lipschitz continuity, Lemmas~\ref{lemma:rollout}--\ref{lemma:ber}, and  Bernstein inequality,  the approximate $\hat{\boldsymbol \theta}_{K+1}$  converges  to its exact value $\boldsymbol \theta_{K+1}$ as  the number of samples and rollouts tend to infinity with a probability greater than $1-({\frac{z}{\varepsilon(L)}})^{-z}$, where $z:=(\sum_{s=1}^S (f(s)+1) d^s_x)(\sum_{s=1}^S (f(s)+1) d^s_u)$.   Subsequently, one gets
$J (\hat{\boldsymbol \theta}_{K+1})-J (\boldsymbol \theta^*)\leq(1-\frac{1}{2}\eta \eta_{max}^{-1})(J (\boldsymbol \theta_K)-J (\boldsymbol \theta^*))$, when $ J (\boldsymbol \theta_K)-J (\boldsymbol \theta^*) \leq \varepsilon(L,T)$. This recursion is contractive, which  is  similar to the proof of Theorem~\ref{thm2}. 
\end{proof}


\begin{remark}
\emph{From~\cite{yang2019provably}, it can be shown that  the policy  gradient algorithms can be extended to actor-critic ones.}
\end{remark}

\section{Numerical examples}\label{sec:numerical}
In this section, we provide two numerical examples.

\textbf{Example 1.}  Consider a risk-sensitive LQ deep structured team with the following parameters: $n=10$, $A=0.9$,   $B=0.4$, $\bar Q=2$,  $R=1$, $Q=1$, $\bar R=1$,  $\eta=5$, $T=10$, $L=100$,  $\alpha^{1:6}= \sqrt{0.5}$,  $\alpha^{7}=\sqrt{1.5}$, 
   $\alpha^{8}= 1$, $\alpha^{9}=\sqrt{2}$,
   $\alpha^{10}=\sqrt{2.5}$,  $w^i_t \sim \text{norm}(0,0.1)$, and $x^i_1 \sim  \text{norm}(0,0.1)$. It is observed in Figure~\ref{fig:PN} that the model-based policy gradient descent algorithm converges to the global optimal solution when the risk factor is $\lambda=0.1$.
  \begin{figure}[t!]
\centering
\scalebox{0.9}{
\includegraphics[ trim={0cm 8cm 0 8cm},clip,width=\linewidth]{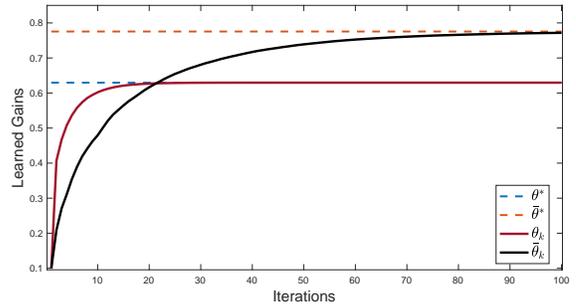}}
\caption{The convergence of the model-based policy gradient descent  algorithm  in Example 1 with a  risk-sensitive cost function.}\label{fig:PN}
\end{figure}

\textbf{Example 2.} Consider  a risk-neutral LQ deep structured teams with  the following parameters: $n=10$, $A=B=Q=1$, $\bar Q=2$,  $R=2$, $\bar R=1$, $\eta=0.2$, $T=10$, $L=100$,  $ \alpha^{1:9}= \sqrt{0.1}$, $\alpha^{10}=\sqrt{9.1}$, $w^i_t \sim \text{norm}(0,0.02)$, and  $x^i_1 \sim  \text{unif}(0,0.1)$.
  \begin{figure}[t!]
\centering
\scalebox{0.9}{
\includegraphics[ trim={0cm 8.2cm 0 8cm},clip,width=\linewidth]{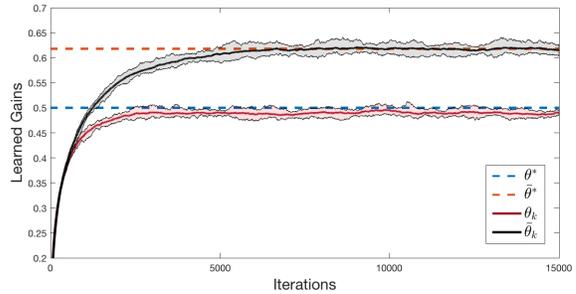}}
\caption{The convergence of the model-free policy gradient descent algorithm  in Example 2 with  the risk-neutral cost function.}\label{fig:PN2}
\end{figure}
The  learning trajectory  of  the model-free policy gradient descent  algorithm is depicted in Figure~\ref{fig:PN2} for $10$ random seeds.  The simulation was run on a 2.7 GHz Intel Core  i5 processor and took roughly $2$ hours.

\section{Implementation}\label{sec:implement}
In practice, agents can use different methods  to implement the  RL algorithms. Below, we mention three types. 

\begin{itemize}
\item \textbf{Team (common)  learner}: All  agents have access to a \emph{common exploration} noise, meaning that  the empirical gradient is  identical for all agents. This way, all agents run the same learning algorithm  with an identical solution, given that  the step sizes are identical for all~agents.  
\item \textbf{Single learner and multiple imitators}. This is when a single agent  learns the optimal  strategy while other agents act as imitators and are passive during the learning process. In particular,  one agent explores the system from its point of view and  others employ the updated (learned) strategy  at each iteration to decide their next actions. It is also possible to select the learner  randomly at each iteration in order to have  a fair implementation.  This type of implementation   is similar to the notion of person-by-person optimality, which is different from the global optimality, in general.  Its advantage over the above  team implementation  is that the single player may use an individualized  observation. For example,  in the natural policy gradient~\eqref{eq:NPGD}, the centralized information $\Sigma_{\theta(s)}$ in~\eqref{eq:risk_neutral_pa} may be replaced by the individualized information $\sum_{t=1}^\infty \Delta x^i_t (\Delta x^i_t)^\intercal$,  because $\Delta x^i_t$  has an identical distribution for all  players in sub-population~$s$.

\item \textbf{Many independent learners}: To avoid discrepancy between the agents during the learning process,   the above implementations allow for  only one common rule of learning. However, if the number of agents is very large,  all agents can independently learn the strategy,  because their explorations  are decoupled from one another. In such a case,  the trajectory of the deep state (the coupling term)  is independent of  i.i.d. exploration~noises.
\end{itemize}

\section{Conclusions}\label{sec:conclusion}
In this paper, we investigated the  convergence of model-based and model-free gradient descent and natural policy gradient descent algorithms in linear quadratic deep structured teams. The size of the parameter space of the proposed algorithms is independent of the number agents in each sub-population,  making the algorithms  applicable to large-scale problems. By using the notions of gradient domination and locally Lipschtiz continuity, we presented an analytical proof for the global convergence of the above algorithms. The theoretical findings were verified by some simulations.  The obtained results naturally extend to other variants of reinforcement learning methods such as  actor-critic.

\bibliographystyle{IEEEtran}
\bibliography{Jalal_Ref}

\end{document}